\documentclass[a4paper,UKenglish,cleveref, autoref, thm-restate]{lipics-v2021}
\nolinenumbers

\pdfoutput=1 
\hideLIPIcs  

\newcommand{\planar}{\ensuremath{P(\Gamma)}}
\newcommand{\planarprime}{\ensuremath{P(\tilde{\Gamma})}}

\newcommand{\restateref}[1]{\hyperref[#1]{$\star$}}
\newcommand{\ch}{\ensuremath{\text{ch}}}

\title{A first view on the density of 5-planar graphs}

 \author{Aaron B{\"u}ngener}{Universit{\"a}t T{\"u}bingen}{aaron.buengener@uni-tuebingen.de}{}{}
 \author{Jakob Franz}
 {Universit{\"a}t T{\"u}bingen}{jakob.franz@student.uni-tuebingen.de}
 {}{}
 \author{Michael Kaufmann}{Universit{\"a}t T{\"u}bingen}
 {michael.kaufmann@uni-tuebingen.de}
 {https://orcid.org/0000-0001-9186-3538}
 {}
 \author{Maximilian Pfister}{Universit{\"a}t T{\"u}bingen}{maximilian.pfister@uni-tuebingen.de}{https://orcid.org/0000-0002-7203-0669}{This research was supported by the DFG grant SCHL 2331/1-1}

\Copyright{Aaron B\"ungener, Jakob Franz, Michael Kaufmann, Maximilian Pfister}

\authorrunning{A.~B{\"u}ngener, J.~Franz, M.~Kaufmann, M.~Pfister} 


\category{Track 1, Long Paper} 

\relatedversion{} 

\EventEditors{}
\EventNoEds{0}
\EventLongTitle{}
\EventShortTitle{}
\EventAcronym{}
\EventYear{}
\EventDate{}
\EventLocation{}
\EventLogo{}
\SeriesVolume{}
\ArticleNo{}

\begin{document}

\maketitle

\begin{abstract}
A key concept for many graph layout algorithms is planarity, a graph property that allows to draw vertices and edges crossing-free in the plane.
Important is the generalization to
$k$-planar graphs, which can be drawn in the plane with at most $k > 0$ crossings per edge. 
One of the basic graph properties that have been explored for those graph classes
is the \emph{maximum edge density}, i.e., the maximum number of edges a $k$-planar graph on $n$ vertices may have. While there are numerous results for the classes of $1$- and $2$-planar graphs, there are few results for increasing $k=3$ or $4$ due to the complex graph structures.
We make a first step towards even larger $k>4$ exploring the class of $5$-planar graphs. 
While our main tool is still a discharging technique, 
a better 
understanding of the structure of the denser parts leads to corresponding density bounds in a much simpler way.

We first apply a simplified version of our technique to outer $5$-planar graphs and surprisingly observe that the structure of
maximally dense (general) $5$-planar graphs
differs from the known uniform structure of maximally dense $k$-planar graphs for smaller $k \in \{1,2,3,4\}$.
As the central result of this paper, we then show that graphs that admit a simple 5-planar drawing have at most $7(n-2)$ edges, drastically improving the previous best bound of $\approx8.3n$. This even implies a small improvement of the leading constant in the Crossing Lemma $cr(G) \ge c \frac{m^3}{n^2}$ from $c=\frac{1}{27.48}$ to  $c=\frac{1}{27.3}$.
To demonstrate the potential of our new technique, we also apply it to 4-planar and 6-planar graphs.

\subparagraph{Generative AI Declaration}
Generative AI was not used in the preparation of this article.

\end{abstract}

\section{Introduction}
In a drawing $\Gamma$ of a graph $G$, the vertices of $G$ are injectively mapped to points in the plane while the edges of $G$ are represented by Jordan arcs that connect the corresponding points\footnote{Refer to Section~\ref{sec:preliminaries} for a complete description.}. The most prominent class of drawings are \emph{planar} drawings in which no two Jordan arcs intersect. The question whether a graph $G$ admits a planar drawing can be answered both combinatorially by identifying forbidden minors~\cite{Kuratowski1930,Wagner1937} or algorithmically (in linear time)~\cite{Boyer2004,Hopcroft1974}. By rearranging Euler's Formula, one obtains a negative answer
as soon as an $n$-vertex graph has more than $3n-6$ edges. Inspired by cognitive experiments concerning the readability of graph drawings~\cite{Mutzel2001,Purchase2000}, the research area of \emph{beyond planarity} emerged in order to categorize the landscape of non-planar drawings. A beyond planar class is usually defined in terms of forbidden crossing configurations for the drawing. One of the most important beyond planarity classes is the family of $k$-planar graphs.
A graph is $k$-planar if it admits a drawing where no edge is crossed more than $k$ times. This extends the class of ($0$-)planar graphs quite naturally. Unfortunately, for $k>0$, the class of $k$-planar graphs is neither minor-closed~\cite{DBLP:journals/siamdm/DujmovicEW17} nor does it admit efficient recognition algorithms~\cite{Urschel2021}.
Hence, a lot of research was devoted to obtaining bounds on the maximum \emph{edge density}, i.e., the maximum number of edges an $n$-vertex graph can have in this class.
The current best upper bounds for $1\leq k\leq 4$ are $4n-8$~\cite{von1983bemerkungen}, $5n-10$~\cite{DBLP:journals/combinatorica/PachT97}, $5.5n-11$~\cite{Pach2006} and $6n-12$~\cite{DBLP:journals/comgeo/Ackerman19}, respectively.
The class of graphs which achieve the corresponding bound are called \emph{optimal} $k$-planar graphs. While for $k\leq2$, there exist both  characterizations~\cite{char23plane, DBLP:journals/algorithmica/GrigorievB07} as well as efficient recognition algorithms~\cite{test1plane,test2plane} for optimal graphs, the results get sparse for increasing $k$. For $k=3$ there still exists a characterization~\cite{char23plane},
while for $k \geq 4$, there is only the pioneering work
of Ackerman on the density of 4-planar graphs using carefully refined discharging techniques 
\cite{DBLP:journals/comgeo/Ackerman19}.

Another motivation to considering edge density bounds for $k$-planar graphs: improvements for small values of $k$ can be used to improve the leading constant of the celebrated Crossing Lemma~\cite{ajtai1982crossing,leighton1983complexity}, which in turn can be used to improve the upper bounds on the maximum edge density for larger values of $k$.
Applied to the case where $k=5$, the current best constant of $\frac{1}{27.48}$~\cite{BestConstantCr} yields a bound of $\approx 8.3n$ for the edge density of $5$-planar graphs.
For $k \leq 4$, $k$-planar graphs have been found whose numbers of edges match the known upper bounds (up to a small constant)~\cite{DBLP:journals/comgeo/Ackerman19,Pach2006,DBLP:journals/combinatorica/PachT97}. Interestingly, all these examples exhibit some remarkable structure. Namely, they contain a set of crossing-free edges $E'$ such that the drawing induced by $E'$ is~a biconnected planar graph whose faces have small size and whose dual graph is simple.\footnote{Holds for $k\in\{3,4\}$ only after removing a constant number of edges from the best known lower bound construction. See also the construction after \cref{thm:edge-bound-polyhedral}.}
We call them \emph{polyhedral $h$-framed drawings}, similar to $h$-framed graphs~\cite{DBLP:journals/combinatorics/BekosLHK24}.

\textbf{Our contribution.}
Inspired by the common structure observed in the lower-bound examples for $k$-planar graphs with $k \leq 4$,
we study the maximum edge density of polyhedral $h$-framed $5$-planar graphs in Section~\ref{sec:outer} and establish that they have at most $6(n-2)$ edges.
To obtain this bound, we show that $n$-vertex outer $5$-planar graphs have at most $4n-9$ edges.
Additionally, we provide a lower-bound construction of a simple $n$-vertex $5$-planar graph with $6.2n-O(1)$ edges. This implies that the trend for the structure of lower-bound graphs, i.e.\ the polyhedral $h$-framed $k$-planar graphs, actually breaks for $5$-planarity.
In Section~\ref{sec:general}, we consider (general) $5$-planar graphs that admit a simple drawing and show that they have at most $7(n-2)$ edges. 
We then use the achieved bound to improve the leading constant of the Crossing Lemma to $\frac{1}{27.3}$. In Section~\ref{sec:further-appl}, we further apply our technique  to show its potential.

Our main idea to achieve all these bounds is to identify critical configurations 
of three pairwise crossing edges, which are required for dense drawings~\cite{Ackerman2007}.
Drawings without such configurations are called \emph{quasi-planar} 
\cite{DBLP:conf/gd/AgarwalAPPS95}.
We isolate these critical parts of the corresponding drawing and analyze them in a unified manner.
This extends an approach of Abrego et al.\cite{abrego2024bookcrossingnumberscomplete}.

\section{Preliminaries}\label{sec:preliminaries}
Throughout this paper we assume that all graphs have $n\ge 3$ vertices and are simple unless otherwise specified.
Let $\Gamma$ be a drawing of a graph in the Euclidean plane with vertices as points and edges as Jordan arcs with the common assumptions that (i) 
any two edges share only finitely many points such that each is either a proper crossing or a common endpoint, (ii) no three edges cross in the same point and (iii) no vertex is an interior point of an edge.
We require that $\Gamma$ is simple in the sense that any two edges share at most one common point, i.e., no two adjacent edges cross and any two edges cross at most once.
We refer to a graph $G$ together with a simple drawing as a \emph{simple topological} graph -- this allows us to not distinguish between the vertices and edges of $G$ and their corresponding points and arcs in $\Gamma$ if unambiguous.
The planar \emph{skeleton} $\phi(G)$ of $G$ in $\Gamma$ is the plane subgraph of $G$ induced by the crossing-free edges of $G$ in $\Gamma$ such that the embedding of
$\phi(G)$ is the one induced by $\Gamma$.
All other edges are \textit{chords} of the face of $\phi(G)$, in which they are drawn.
Graph $G$ is called \emph{$h$-framed} if 
$\phi(G)$ is biconnected, spans all vertices of $G$ and its faces have size at most~$h$.
A drawing of graph $G$ is called \emph{polyhedral} $h$-framed if it is $h$-framed and, in addition, $\phi(G)$ is triconnected. So, any two faces of $\phi(G)$ share at most one edge, hence the dual graph (w.r.t.~$\phi(G)$) is simple.
A drawing of $G$ is called \emph{outer} if all~vertices of $G$ are incident to the same (unbounded) biconnected 
face $f_o$ of $\phi(G)$ and no edge intersects the interior of~$f_o$.

\section{Outer 5-planar and polyhedral h-framed 5-planar graphs}
\label{sec:outer}
In order to derive the main result of this section, i.e., \cref{thm:edge-bound-polyhedral}, we first establish an upper bound on the edge density of outer $5$-planar graphs. Note that the previous best upper bound for $n$-vertex outer $5$-planar graphs was $\approx 5.5n$ due to~\cite[Thm.~21]{DBLP:conf/compgeom/AichholzerOOPSS22}.
\begin{figure}[t]
    \centering
    \hspace*{\fill}
     \begin{subfigure}{0.26\linewidth}
        \includegraphics[width=\linewidth, page=1]{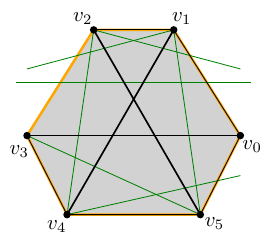}
        \subcaption{}
        \label{sfig:outer-nomenclature}
    \end{subfigure}
    \hfill
    \begin{subfigure}{0.26\linewidth}
        \includegraphics[width=\linewidth, page=2]{outer-planar}
        \subcaption{}
        \label{sfig:outer-l=1}
    \end{subfigure}
    \hspace*{\fill}
    \caption{Details of the proof of \cref{thm:outer-5-planar}. (a) The edge sets $C$ (heavy), $X$ (green) and the region $\cal R$ (gray) for $l\ge3$ and (b) the case $l=1$.}
    \label{fig:placeholder}
\end{figure}
\begin{theorem}\label{thm:outer-5-planar}
    An $n$-vertex outer $5$-planar graph has at most $4n-9$ edges.
\end{theorem}
\begin{proof}
We prove the theorem by induction on $n$. 
Let $G=(V,E)$ be a graph with outer 5-planar drawing $\Gamma$. 
Assume w.l.o.g. that $\Gamma$ is a straight-line drawing of $G$ such that the vertices are in general position and form the corners of a convex $n$-gon \cite[Thm.~3]{DBLP:journals/dam/HongN19}.
By definition, $\Gamma$ has a face $f_o$ which contains all vertices of~$G$. 
Since $\binom{n}{2} \leq 4n-9$ for $n \leq 6$, the statement trivially holds for $n \leq 6$. Let $n \ge 7$.
If $\Gamma$ does not contain three pairwise crossing chords, then $G$ is an outer quasi-planar graph which has at most $4n-10$ edges~\cite{quasi-planar}. Hence, assume
that there exist three pairwise crossing chords $C= \{c_1,c_2,c_3\}$ in $\Gamma$, called \textit{$C$-edges}.
Let $X \subseteq E \setminus C$ be the set of chords which cross at least one of $c_1$, $c_2$, or $c_3$ and call them \textit{$X$-edges}.
Since $c_1,c_2,$ and $c_3$ pairwise cross and since each chord has at most five crossings, it follows that $|X| \leq 9$. Note that if there is an $X$-edge that crosses more than one $C$-edge, then $|X| < 9$.
The convex hull of $C$ defines a region $\mathcal{R}$. See \cref{sfig:outer-nomenclature} for an example.
The crucial observation is that only chords of $X$ and $C$ intersect the interior of~$\mathcal{R}$, while at most six edges -- which may coincide with edges of $f_o$ --  lie on the boundary of $\mathcal{R}$. 
Remove $X$ and $C$ from $G$ and add any missing edge on the boundary of $\cal R$ to simplify the construction. The resulting graph $G'$ has at most $|X| + |C| \le 12$ edges less than $G$.

Denote the set of the boundary edges of $\cal R$ by $B$.
Let $l \le 6$ be the number of chords of $\Gamma$ in~$B$. We can split $G'$ along these chords to get $\cal R$ and $l$ other outer 5-planar graphs $G_i=(V_i, E_i)$ on $3 \le n_i < n$ vertices ($i\in [1,l]$).
The graphs $\cal R$ and $G_i$ share exactly one edge, which is the corresponding chord of $\Gamma$ in $B$.
By induction, $|E_i| \le 4n_i - 9$ edges.

Consider $l = 1$, and let w.l.o.g. the unique chord of $\Gamma$ in $B$ be $v_0v_1$ (see \cref{sfig:outer-l=1}). Observe that $v_2v_5\in C$ is not incident to any vertex of $V_1$. As such, at most two $X$-edges cross no other $C$-edge than $v_2v_5$, namely $v_0v_4$ and $v_1v_3$. Either $v_2v_5$ is not crossed by three $X$-edges or at least one $X$-edge crossing $v_2v_5$ crosses another $C$-edge. In both cases $|X| \le 8$. As $\cal R$ and $G_1$ share an edge of $B$ we have $|E| \le |C|+|X| + |B| + |E_1| - 1 \le 3 + 8 + 6 + 4 (n-4) - 9 - 1 = 4n - 9$.
 
For the other case $l \ge 2$, we split $G'$ into $l+1$ parts along the chords in $B$.
Summing over the different $G_i$, we get
$n = 6 + \sum_{i=1}^l (n_i - 2) = \sum_{i=1}^l n_i + 6 - 2l.
$
The number of edges in $G$ is at most 
\begin{align*}
    |E| &\le |C| + |X| + |B| + \sum_{i=1}^l (|E_i| - 1) \le 3 + 9 + 6 + \sum_{i=1}^l (4n_i - 10)\\
    &\le 4(n - 6 + 2l)  + 18 - 10l = 4n - 2l - 6,
\end{align*}
which is at most $4n-10$ for $l\ge 2$.
\end{proof}

For a lower bound construction, choose $n = 10x+2$ for some integer $x\geq 1$, split an $n$-gon into $x$ faces of size $12$ and insert $26$ chords into every face (which is possible as witnessed by \cref{fig:optimal-dodecagon}). This results in the following lower bound.
\begin{theorem}\label{thm:lower-outer}
There exist outer $5$-planar graphs with $n$ vertices and $3.7n-6.4$ edges.
\end{theorem}

We can use \cref{thm:outer-5-planar} to improve the constant of the Crossing Lemma that is specifically  tailored to 
the convex setting~\cite{DBLP:conf/compgeom/AichholzerOOPSS22}.

\begin{theorem}
\label{lem:outer-crossing-lemma}
Let $G$ be a graph with $n$ vertices and $m$ edges such that $m \geq \frac{73}{16}n$. For any outer drawing of $G$ the crossing number $cr_o(G)$ is at least
$ \frac{1}{10.41}\frac{m^3}{n^2}$.
\end{theorem}

\begin{proof}
At first, we achieve a new linear bound for the number of crossings in $G$.
    
Let $\Gamma$ be a crossing-minimal outer drawing of $G$. We iteratively delete one of the most crossed edges from $\Gamma$ as long as there are more than $4n-9$ edges. By \cref{thm:outer-5-planar}, these edges have at least six crossings. We account at least $6(m - (4n-9))$ removed crossings. In the next phase (between $4n-9$ and $3.5n -6$ edges), we account $5$ crossings per deleted edge. Continuing this process using the best bounds for $k\le 4$~\cite{abrego2024bookcrossingnumberscomplete}, we obtain
\begin{align*}
cr_o(G) \geq \,\, & 6\left(m - (4n-9)\right) + 5\left((4n-9) - \left(3.5n-6\right)\right) + \\
&4\left(\left(3.5n-6\right) - (3.25n-5.5)\right) + 3\left((3.25n-5.5)-\left(3n-5\right)\right) + \\ &2\left(\left(3n-5\right) - (2.5n-4)\right) +\left(\left(2.5n-4\right) - (2n-3)\right)  = 6m - 18.25n + O(1)
\end{align*}

We proceed following the common probabilistic arguments from \cite{DBLP:books/daglib/0019107}. Fix $p = \frac{73n}{16m} \le 1$ and select independently each vertex from $G$ with probability $p$. For the subdrawing $\Gamma'$ of the induced random subgraph $G'$, we have $\mathbb{E}[n'] = pn$, $\mathbb{E}[m'] = p^2m$ and $\mathbb{E}[cr_o(G')] = p^4 cr_o(G)$ for the number of vertices, edges and crossings. By linearity of expectation, we derive from the bound above that $\mathbb{E}[cr_o(G')] \ge 6 \mathbb{E}[m'] - 18.25 \mathbb{E}[n']$ holds. Plugging everything in yields
\[cr_o(G) \geq \frac{6m}{p^2} - \frac{18.25n}{p^3}=
\frac{512m^3}{5329n^2} \geq \frac{1}{10.41}\frac{m^3}{n^2}.\qedhere\]
\end{proof}

The following corollary can immediately be derived from \cref{thm:outer-5-planar}. 
\begin{corollary}\label{cor:chords-per-face}
    Let $f$ be a biconnected face and let $|f|$ be the number of vertices of $f$.
    If every internal chord of $f$ is crossed at most five times, then $f$ contains at most 
    $3|f|-9$ chords.
\end{corollary}
Although the bounds of Theorems \ref{thm:outer-5-planar} and \ref{thm:lower-outer} are not tight, 
\cref{cor:chords-per-face} is sufficient to de\-rive a bound for polyhedral $h$-framed $5$-planar graphs, which is tight up to an additive constant. 
\begin{theorem}\label{thm:edge-bound-polyhedral}
An $n$-vertex $5$-planar polyhedral $h$-framed graph has at most $6n-18$ edges.

\end{theorem}
\begin{proof}
Let $G$ be a $5$-planar polyhedral $h$-framed graph with $n$ vertices and $m$ edges. Let $F$ be the set of faces and let $f_i$ be the number of faces of size $i$ of $\phi(G)$. We can bound $m$ with 
\begin{equation}\label{eq:poly-framed-first}
     m \leq 1.5f_3+4f_4+7.5f_5+\sum_{i=6}^h (3.5i-9)f_i
 \end{equation}
 since the contribution of each face to the total number of edges is half its size (as boundary edges will be counted twice) in addition to the maximum number of chords inside each face such that every chord is crossed at most five times. For $i \leq 5$ these are all possible $\frac{i(i-3)}{2}$ chords, while for $i \geq 6$ we use the upper bound of Corollary \ref{cor:chords-per-face}. We can further observe that
\begin{equation}\label{eq:poly-framed-second}
       \sum_{i=3}^h (i-2)f_i = 2n-4\
\end{equation}
holds since, when triangulating each face, the total number of triangles in the plane is $2n-4$, while a face of size $i$ accommodates $i-2$ triangles as $\phi(G)$ is biconnected.

Moreover, the dual of $G$ is plane (as $\phi(G)$ is plane) and simple (as $G$ is polyhedral). We apply Euler's formula to get $\sum_{i=3}^h i\cdot f_i \leq 6|F|-12$. Rearranging yields
\begin{equation}\label{eq:poly-framed-third}
\sum_{i=3}^h (i-6)f_i \leq -12.
\end{equation}

Multiply both sides of (\ref{eq:poly-framed-third}) by $-\frac{1}{2}$ and add the inequality to (\ref{eq:poly-framed-first}) to obtain
\begin{equation*}
    m + 6 \leq 3f_3+5f_4+8f_5+\sum_{i=6}^h (3i-6)f_i \le 3\cdot \left(\sum_{i=3}^h (i-2)f_i\right) \stackrel{(\ref{eq:poly-framed-second})}{\le} 6n-12.\qedhere
\end{equation*}

\end{proof}

A nearly tight lower bound can be obtained by a small adjustment of Ackerman's cylindrical construction for simple 4-planar graphs \cite{DBLP:journals/comgeo/Ackerman19}. We start with a hexagonal tiling (see \cref{fig:hex-skeleton}) consisting of $0.5(n-2)$ hexagons and $1.5(n-2)$ edges. To ensure that the planar skeleton is tri-connected, we add two edges each at the top and the bottom face such that we create two triangles and one quadrilateral each.
We insert nine edges in the other hexagons and can add two more edges in each of the two quadrilaterals without creating any multi-edges. Adding everything up yields
$1.5(n-2) + (0.5(n-2)-2) \cdot 9 + 8 = 6n-22$ edges.

The following construction, which is based on \cite{DBLP:journals/comgeo/Ackerman19,Pach2006}, establishes 
that optimal $5$-planar graphs have linearly more edges than polyhedral $h$-framed graphs, breaking the pattern set by $k\le 4$, where the edge densities differed only by a constant number of edges.

\begin{theorem}
  \label{thm:simple-framed-lb}
 There exist simple $5$-planar graphs with $n$ vertices and $6.2n - 18.4$ edges.
\end{theorem}
\begin{figure}[t]
    \centering
    \begin{subfigure}[b]{.10\textwidth}
		\centering
		\includegraphics[width=\textwidth,page=1]{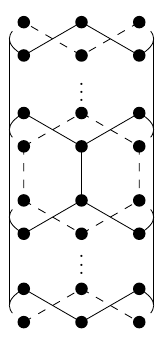}
		\subcaption{}
		\label{fig:hex-skeleton}
	\end{subfigure}
 	\hfil
	\begin{subfigure}[b]{.22\textwidth}
		\centering
		\includegraphics[width=\textwidth,page=1]{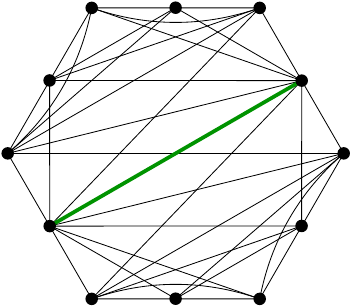}
		\subcaption{}
		\label{fig:optimal-dodecagon}
	\end{subfigure}
    \hfill
 	\begin{subfigure}[b]{.18\textwidth}
		\centering
		\includegraphics[width=\textwidth,page=2]{simple-optimal-skeleton.pdf}
		\subcaption{}
		\label{fig:dodecagon-rotation}
	\end{subfigure}
 	\hfil
 	\begin{subfigure}[b]{.22\textwidth}
		\centering
		\includegraphics[width=\textwidth,page=2]{dodecagon_optimal.pdf}
		\subcaption{}
		\label{fig:top-dodecagon}
	\end{subfigure}
     	\hfil
 	\begin{subfigure}[b]{.22\textwidth}
		\centering
		\includegraphics[width=\textwidth,page=3]{dodecagon_optimal.pdf}
		\subcaption{}
		\label{fig:bottom-dodecagon}
	\end{subfigure}
    \label{fig:simple-construction}
 	\caption{(a) $H_x$, the hexagonal tiling of a cylinder (based on an illustration in \cite{DBLP:journals/comgeo/Ackerman19}), (b) an optimal outer 5-planar dodecagon (found by computer), (c) rotation of dodecagons on the cylinder and (d),(e) $f_t, f_b$ with their neighboring faces. There, only the green diagonal and the chords, which may create multi-edges, are drawn.}
  \end{figure} 
\begin{proof}
    
Construct the topological graph $H_x$ as a hexagonal tiling of the surface of a cylinder consisting of $x$ layers of three hexagons wrapped around the cylinder (see \cref{fig:hex-skeleton}). Note that $H_x$ consists of $6x + 6$ vertices, $9x + 6$ edges, and $3x + 2$ faces of size six (one sitting at the top and bottom of the cylinder, respectively).
Construct a dodecagonal tiling $D_x$ of a cylinder surface by adding one vertex on every edge of $H_x$. Call those vertices $V_D$ and the other vertices $V_H$. Graph $D_x$ has $n = 15x + 12$ vertices and $18x + 12$ edges in the skeleton $\phi(D_x)$. Ignore the top and bottom faces $f_t$ and $f_b$ of $D_x$ for now. Fill each of the $3x$ 12-gons on the lateral surface of $D_x$ with $26$ chords (\cref{fig:optimal-dodecagon}). 
This will create (non-homotopic) multi-edges. Each of these is a pair of parallel edges enclosing one vertex of $V_D$. 
To minimize the number of multi-edges rotate each filled $12$-gon such that the diagonal marked green in \cref{fig:optimal-dodecagon} is incident to two vertices of $V_D$ and that no vertex of $V_D$ is incident to two green diagonals. This is possible, as shown in \cref{fig:dodecagon-rotation}. Note that a vertex of $V_D$ that is incident to a green diagonal is not enclosed by parallel edges, as the edge connecting its two neighboring vertices does not exist in the face of the green diagonal, refer to \cref{fig:optimal-dodecagon}.
    
$9x - 6$ of the vertices in $V_D$ are not incident to $f_t$ or $f_b$ and therefore may be enclosed by parallel edges. The $3x$ green diagonals are incident to $6x$ vertices in $V_D$, of which two are incident to $f_t$ and $f_b$ each (see \cref{fig:dodecagon-rotation}). Therefore, there are $9x-6 - (6x - 4) = 3x - 2$ pairs of parallel edges left. From each such pair remove one edge to obtain a simple graph.
	
Fill $f_t$ and $f_b$ with $21$ chords each according to \cref{fig:top-dodecagon,fig:bottom-dodecagon}, depending on the orientation of the green diagonals in the neighboring faces. If necessary, mirror the neighboring faces along the green diagonal so that their chords connecting two vertices of $f_t$ or $f_b$ are exactly the ones shown in  \cref{fig:top-dodecagon,fig:bottom-dodecagon}. This does not change any previous step.
	
Finally, count the edges. The skeleton $\phi(D_x)$ has $18x+12$, the faces on the lateral surface have $3x \cdot 26 - (3x - 2)$ and the top and bottom face each have $21$ edges. In sum, $D_x$ has $93 x + 56$ edges. With $x = \frac{n -12}{15}$ we get an edge density of $6.2n - 18.4$.
\end{proof}

\section{General 5-planar graphs}\label{sec:general}

In order to derive our main result, we will use  \emph{discharging techniques} similarly as in \cite{Ackerman2007}.
We denote by $\planar$ the so-called \emph{planarization} of $\Gamma$, i.e., the vertices and crossing points of $\Gamma$ are the vertices of $\planar$, while the edges of $\planar$ are the crossing-free segments in $\Gamma$ which are bounded by vertices and crossings. The faces of $\planar$ are the con\-nected regions bounded by the edges of $\planar$.
The vertices of $\planar \cap G$ are called \emph{original}.

\begin{theorem}\label{thm:5planar}
An $n$-vertex simple topological $5$-planar graph $G$ has at most $7(n-2)$ edges.
\end{theorem}
\begin{proof}
Modify a simple $5$-planar drawing $\Gamma$ of $G$ and construct a drawing $\tilde{\Gamma}$ (of a graph $\tilde{G}$)
whose planarization $\planarprime$ does not contain triangular faces with three crossings on their boundary. We call these faces $0$-triangles (see the general definition of $x$-$y$-faces given below).

The drawing $\tilde{\Gamma}$ facilitates counting the edges of $G$. More in detail, we highlight that we put charges to the faces of $\planarprime$ and redistribute them to the edges of original graph $G$ to count the latter ones.
We start with a description how to obtain $\tilde{\Gamma}$ from $\Gamma$. In doing so, we ensure that every edge deleted from or added to $\Gamma$ will be mapped
to a specific configuration~-- an empty hexagonal region $H_i$ -- in $\tilde{\Gamma}$. Thus, we can take those edges, i.e., $E^-(H_i)$, $E^+(H_i),\dots$, to be defined below, 
into account for the discharging on $\tilde{\Gamma}$ later.

\medskip

Let $C_\triangle$ be the set of all 0-triangles in $\planar$, formed by segments of three pairwise crossing edges of $\Gamma$. 
Choose one such $t_1 \in C_\triangle$ (if it exists) and let $e_1$, $e_2$, and $e_3 \in \Gamma$ be the edges whose segments form $t_1$. 
The endpoints of $e_1$, $e_2$, and $e_3$ are disjoint since $\Gamma$ is simple. Let $H_1$ denote a minimal region that contains the three edges defining $t_1$ such~that its boundary is simple and does not create an empty lens with an edge of $\Gamma$ (i.e., a region bounded by two edge segments without an enclosed vertex).
For an illustration, see \cref{fig:H-definition-a}. 

\begin{figure}[t]
    \centering
    \begin{subfigure}[b]{.30\textwidth}
	\centering    \includegraphics[width=0.9\textwidth,page=1]{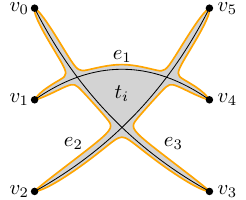}
	\subcaption{}
    \label{fig:H-definition-a}
	\end{subfigure}
 	\hfill
    \centering
    \begin{subfigure}[b]{.3\textwidth}
	\centering    
    \includegraphics[width=0.9\textwidth,page=2]{H-region.pdf}
	\subcaption{}
    \label{fig:H-definition-b}
	\end{subfigure}
 	\hfill
	\begin{subfigure}[b]{.30\textwidth}
	\centering \includegraphics[width=0.9\textwidth,page=3]{H-region.pdf}
	\subcaption{}
    \label{fig:H-definition-c}
	\end{subfigure}
    \hfill
	\caption{(a) The region of an H-block (gray) and its boundary (orange) for the base case that no other edge is intersecting $e_1,e_2,e_3$. (b) Different kinds of edges in $\Gamma$: $e_1,e_2,e_3$, $v_0v_2, v_0v_4,v_1v_5, v_3v_5$ and the two green edges belong to $E^-(H_i)$.
    For constructing $\tilde\Gamma$, all these edges are deleted, while $v_0v_5, v_1v_2, v_3v_4 \in E^+(H_i)$ are added. For the corresponding hexagonal face $f_i$ in $P(\tilde\Gamma)$ in (c), the two green edges from (b) are in $E_{HQ}(f_i) \cup E_{HH}(f_i)$ and the other eight dele\-ted edges are in $E_H(f_i)$.
    Note that there might be other elements of $C_\triangle$ to be chosen instead of~$t_i$.}
		\label{fig:H-definition}
\end{figure}

We remove any edge which intersects the interior of $H_1$ (i.e., the edges that intersect $e_1$, $e_2$, or $e_3$) making the boundary of $H_1$  crossing-free and denote these edges by $E^-(H_1)$.
Further, we add boundary edges of $H_1$ (orange in \cref{fig:H-definition-c}), if they do not exist in $\Gamma$, and denote them by $E^+(H_1)$. The resulting drawing is denoted by $\Gamma_1$. 
Observe that by removing edges that intersect the interior of $H_1$, we destroy $t_1$ and possibly other elements of $C_\triangle$ (if an edge of $E^-(H_1)$ is part of another $t\in C_\triangle$).
Still, this aligns with our goal to remove all 0-triangles while every deleted edge from $\Gamma$ can be assigned to one hexagonal region $H_i$.

We iteratively repeat this process on the remaining elements of $C_\triangle$, denoted by $t_2,t_3,\dots,t_k$, until we obtain a  drawing $\Gamma_k =: \tilde{\Gamma}$ with corresponding planarization $P(\tilde{\Gamma})$, such that $C_\triangle = \emptyset$.

By construction, the interiors of $H_i$ and $H_j$ are disjoint for any $i \neq j$ (they can share boundary edges). We will call these resulting empty hexagonal faces of $\tilde{\Gamma}$ \emph{H-blocks}.
Any connected region of the remainder of $\tilde{\Gamma}$ is called a \emph{Q-block}\footnote{H-blocks originate from ``hexagons'', while Q-blocks are (in a sense) ``quasiplanar''.}.

Let $f_i$ be the hexagonal face in $P(\tilde{\Gamma})$ corresponding to the region $H_i$.
We define $E^-(f_i) := E^-(H_i)$ and $E^+(f_i) := E^+(H_i)$ for all $i\le k$.
Further, we use the notation $E^- := \bigcup_{i\le k} E^-(f_i)$ and $E^+ := \bigcup_{i\le k} E^+(f_i)$.
Now, we can distinguish between three types of edges in $E^-(f_i)$: Edges of $E_H(f_i)$ are contained completely in $H_i$ and edges of $E_{HH}(f_i)$ intersect $H_i$ and $H_j$ for $i \ne j$. Finally, edges of $E_{HQ}(f_i)$ intersect an $H_i$ and a Q-block, but do not intersect another $H_j$ if $i \ne j$.
With this, we have $E^-(f_i) = E_H(f_i) \cup E_{HH}(f_i) \cup E_{HQ}(f_i)$.
For an illustration of these definitions, see \cref{fig:H-definition-b}.

\medskip

\textbf{Charging:}
Let $V(\tilde{\Gamma})$, $E(\tilde{\Gamma})$, and $F(\tilde{\Gamma})$ denote the vertex, edge, and face sets of $\planarprime$, respectively. Denote by $c(\tilde{\Gamma})$ the number of connected components of $\tilde{\Gamma}$.
The boundary of a face $f$ in $\planarprime$ consists of all the vertices and edges of $\planarprime$ incident to $f$.
Note that after deleting edges from $\Gamma$, $\tilde{\Gamma}$ is not necessarily connected. Thus, we define the size $|f|$ of a face $f$ as the number of edges on its boundary counted while traversing all connected components of its boundary.
Analogously, $v(f)$ is defined to be the number of original vertices on the boundary of $f$ that are not isolated, counted with multiplicity while traversing all connected components of its boundary. We call an edge $e \in E(\tilde{\Gamma})$ an $x$-edge, if $e$ is incident to $x$ many original vertices.
We will refer to $f$ as an $x$-$y$ face if $|f| = y$ and $v(f) = x$.
For the cases $y = 3,4,5,6$ and $x$ fixed, we also say $x$-triangle, $x$-quadrilateral, $x$-pentagon and $x$-hexagon.
We define the number of additional components (i.e., the number of components minus one) of the boundary as $c(f)$.
Clearly, $\sum_{f \in F(\tilde{\Gamma})} v(f) = \sum_{v \in V(\tilde{G})} \deg(v)$ and
$\sum_{f \in F(\tilde{\Gamma})} |f| = 2|E(\tilde{\Gamma})| = \sum_{u \in V(\tilde{\Gamma})} \deg(u)$ holds.
Every vertex in $V(\tilde{\Gamma}) \setminus V(\tilde{G})$ is a crossing point in $\tilde{G}$
and therefore its degree in $\planarprime$ is four. Hence,
\[
    \sum_{f \in F(\tilde{\Gamma})} v(f) = \sum_{v \in V(\tilde{G})} \deg(v) =
        \sum_{u \in V(\tilde{\Gamma})} \deg(u) - \sum_{u \in V(\tilde{\Gamma}) \setminus V(\tilde{G})}\deg(u) =
    2|E(\tilde{\Gamma})| - 4\left(|V(\tilde{\Gamma})|-n\right)
\]

Assigning every face $f \in F(\tilde{\Gamma})$ a charge of $\ch(f) = |f|+v(f)-4+4 c(f)$, we get in total:
\begin{align*}
    \sum_{f \in F(\tilde{\Gamma})}&
    \left(|f|+v(f)-4 + 4c(f)\right)\\ &= 2|E(\tilde{\Gamma})|+ 2|E(\tilde{\Gamma})| - 4\left(|V(\tilde{\Gamma})|-n\right) - 4|F(\tilde{\Gamma})| +4(c(\tilde{\Gamma}) -1)\\
    &= 4n + 4(|E(\tilde{\Gamma})|  - |V(\tilde{\Gamma})|  - |F(\tilde{\Gamma})| +c(\tilde{\Gamma})-1) \\
    &= 4n-8,
\end{align*}

using that $\sum_{f \in F(\tilde{\Gamma})} c(f) = c(\tilde{\Gamma}) -1$ holds and applying Euler's Formula for not necessarily connected planar graphs.

\medskip

\textbf{Discharging:}
We will redistribute the charge in several steps over the faces and edges such that the charge of every face of $\tilde{\Gamma}$ is non-negative
and the charge of every edge in $G$ is at least $2\alpha$ for $\alpha := \frac{2}{7}$.
Then, we will obtain
\[4n-8 \geq 2\alpha|E(G)| \Leftrightarrow |E(G)| \leq \frac{4n-8}{2\alpha} = 7(n-2).\]

We give an overview of the discharging steps.
$\ch_i(f)$ will denote  the charge of a face $f$ after step $i$. 
First, we define a supplier-receiver relation between faces with positive initial charge and 1-triangles, the so-called wedge- and side-relations.

\begin{figure}[t]
    \centering
    \hspace*{\fill}
    \begin{subfigure}[b]{.450\textwidth}
	\centering    \includegraphics[width=0.6\textwidth,page=2]{neighbors.pdf}
	\subcaption{}
    \label{fig:neighbors-a}
	\end{subfigure}
 	\hfill
    \centering
    \begin{subfigure}[b]{.45\textwidth}
	\centering    
    \includegraphics[width=0.6\textwidth,page=1]{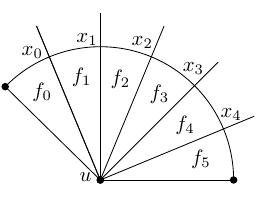}
	\subcaption{}
    \label{fig:neighbors-b}
	\end{subfigure}
 	\hspace*{\fill}
	\caption{(a) The wedge-supplier of $f_0$ at edge $x_0y_0$ is $f_3$. (b)
    Definition of the (first) side-receiver of $f_0$ at edge $ux_1$, which is $f_1$ here. Analogously, $f_2$ is a second side-receiver of $f_0$.  Further, 2-triangle $f_4$ has $f_4$ and $f_3$ as have first and second side-receiver at edge $ux_4$.}
		\label{fig:neighbors}
\end{figure}

\smallskip\noindent\textbf{Wedge-relation} (see \cref{fig:neighbors-a}).
Let $f=f_0$ be any face with two consecutive crossing points $x_0, y_0$ in $\tilde{\Gamma}$ connected by an edge $e_0 = x_0y_0$, and let $f_1$ be the
immediate neighbor of $f=f_0$ at $x_0y_0$. For $i = 1,2,3...$, proceed as follows:
If $f_i$ is neither a 0-quadrilateral nor a 1-triangle, then $f$ has no wedge-relation at edge $e_0$.
If $f_i$ is a 0-quadrilateral with vertices $x_{i-1}, y_{i-1}, y_{i}, x_{i}$ in clockwise order, then we consider the edge $e_{i}=x_{i}y_{i}$ and its immediate neighbor $f_{i+1}$. We increment $i$ and iterate along 0-quadrilateral until we find a neighboring non-0-quadrilateral $f_j$. The search succeeds if $f_j$ is a 1-triangle, and ends with a failure otherwise. Note that the search might end as soon as $f_j=f_1$. $f_j$ is called the wedge-receiver of $f$ at edge $e_0=x_0y_0$. Symmetrically, $f$ is called the wedge-supplier of the 1-triangle $f_j$ (at edge $x_{j-1}y_{j-1}$). 

Similar to the wedge-relation for 0-edges, we define a side-relation for face $f$ at 1-edges.

\smallskip\noindent\textbf{Side-relation} (see \cref{fig:neighbors-b}).
Let $f=f_0$ be any face with a 1-edge, i.e. an edge $e_0 = ux_0$ with one original vertex $u$ and a crossing point $x_0$ in $\tilde{\Gamma}$. Let $f_1$ be the immediate neighbor of $f=f_0$ at $ux_0$.
If $f_1$ is a 1-triangle, it is called first side-receiver of $f$ at $ux_1$. If not, side-receiver of $f$ is not being defined. In the first case, let $ux_1, x_1 \neq x_0$ be the other 1-edge of $f_1$, with an adjacent face $f_2$ with $f_2 \neq f_1$. If $f_2$ is a 1-triangle as well, it will be called second side-receiver of $f$ at $ux_0$. Otherwise, there is no second side-receiver of $f$ at $ux_0$. 

Note that we could have defined a third or fourth side-receiver of $f$ at $e_0$ as well, but we will see that it is sufficient to consider first and second side-receiver of $f$ (at edge $ux_0$)

\begin{itemize}
    \item Step 1: Every 1-triangle receives $\frac15$ charge from its wedge-supplier, see \cref{fig:neighbors-a}.
    \item Step 2:
    Every face $f$ with a 1-edge $ux_1$ contributes $\alpha - \frac15$ charge to its first and second side-receivers,
    except in the case that $f$ is a
    \begin{itemize}
        \item 1-triangle, which does not contribute charge to their side-receiver,
        \item 1-quadrilateral, which only contributes to its first side-receiver (\cref{fig:charge-1-quad}), and
        \item 2-triangle $uxv$, which is traversed by one or two edges of $E^-$, that cross say $uv$ and $vx$, contributes only either to its first side-receiver or to none  at all (\cref{fig:charge-2-triag-1}. 
    \end{itemize}

    \item Step 3: Every face $f$ contributes $\alpha \cdot v(f)$ charge to the edges in $(E\cup E^+) \setminus E^-$, that is $2\alpha$ per edge. The edges in $E^+$ pass their charge over to their adjacent 6-hexagons.
    \item Step 4: Every $6$-hexagon $f$ contributes $2\alpha$ charge to the edges in $E_H(f)$ and $E_{HQ}(f)$. It further contributes $\alpha$ charge to edges in $E_{HH}(f)$.
    (Note that $E_{HQ}(f), E_{H}(f)$ and $E_{HH}(f)$ may be empty.)
    \item Step 5: Each face $f$ 
    distributes its positive charge over the neighboring $6$-hexagons.
\end{itemize}

Before we prove the full correctness of the discharging, we start with
some observations:

\begin{proposition} 
Each 1-triangle has a corresponding wedge-supplier paying in Step 1.
    
\end{proposition}

\begin{proof}
For contradiction, assume that there is a 1-triangle $uxy$ without any wedge-supplier. We go backwards from $uxy$ through the (possibly empty) chain of 0-quadrilaterals starting from $xy$ until we reach a non-0-quadrilateral $f$. By 5-planarity, this happens after a chain of at most four 0-quadrilaterals.
Clearly $f$ is neither a 0-triangle nor a 1-triangle, because of our simplicity assumptions.
Hence $f$ is either a 1-quadrilateral, 2-triangle, 0-pentagon, or of even larger size, and it can pay the charge.
\end{proof}

\begin{figure}[t]
        \centering
        \hspace*{\fill}
    \begin{subfigure}[b]{.24\textwidth}
		\centering
		\includegraphics[width=1\textwidth,page=1]{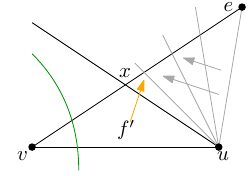}
		\subcaption{}
		\label{fig:charge-2-triag-1}
	\end{subfigure}
     \hfill
    	\begin{subfigure}[b]{.24\textwidth}
		\centering
		\includegraphics[width=1\textwidth,page=2]{five-planar-cases.pdf}
		\subcaption{}
		\label{fig:charge-2-triag-2}
	\end{subfigure}
    \hfill
    \begin{subfigure}[b]{.24\textwidth}
\centering\includegraphics[width=1\textwidth,page=4]{five-planar-cases.pdf}
		\subcaption{}
		\label{fig:charge-1-quad}
	\end{subfigure}
    \hspace*{\fill}
 	\caption{(a)-(b) Discharging of a $2$-triangle $f'$ for the cases that it can contribute to less than two side-neighbors via $ux$. Green edges belong to $E^-$. (c) Discharging of a $1$-quadrilateral. In all subfigures, the part shown in gray represents one of the possible configurations.}
    \label{fig:charge-details}
  \end{figure}
\begin{proposition} 
Each 1-triangle has a corresponding side-supplier paying in Step 2.    
\end{proposition}
\begin{proof}

Let $f$ be a 1-triangle $uxy$. 
By 5-planarity, $f$ is a first or a second side-receiver of some face $f'$, which pays, if it is not a 2-triangle or a 1-quadrilateral.
Assume first that $f'$ is a 2-triangle $uvx$ and side-supplier to $f$ at the edge $ux$, where $u$ denotes an original vertex and $x$ the crossing of $f'$.
The reason for $f'$ to not contribute charge to $f$ in step 2 is that one or even more  edges of $E^-$ intersect $vx$, see \cref{fig:charge-2-triag-1,fig:charge-2-triag-2}.
Thus, the edge $e \in \tilde\Gamma$ starting at $v$ and with first crossing $x$ must be crossed at most four or even three times in $\tilde\Gamma$ and there are  at most three or two resp. 1-triangles at $u$ implying that $f$ is the first side-receiver of $f'$ or even a second side-receiver via edge $uy$ from its other side.

Now, let $f'$ be a 1-quadrilateral incident to vertex $u$ with edge $ux$. There is an additional crossing at $f'$, and by 5-planarity, there are at most three 1-triangles incident to $u$. So, it is sufficient, if $f'$ sends charge only to its first side-receiver, see~\cref{fig:charge-1-quad}.
\end{proof}

Step 4 ensures that each edge of $E^-$ receives $2\alpha$ charge (edges in $E_{HH}(f)$ receive another $\alpha$ charge from a second face $f'$). 
Step 3 guarantees the same for the edges in $E \setminus E^-$ as
\[
\alpha \sum_{f \in F(\tilde{\Gamma})} v(f) = \alpha \sum_{v \in V(\tilde{G})} \deg(v) = 2\alpha \cdot \vert (E \cup E^+) \setminus E^-\vert
\]
holds (while edges in $E^+$ do not keep their charge).

Hence, we only need to show that the final charge $\ch_5(f)$ of all faces $f \in \tilde{\Gamma}$ is nonnegative.
As in the last step only excesses are contributed, this is in the most cases already implied by $\ch_4(f) \ge 0$ and for all faces except $6$-hexagons, even $\ch_3(f) \ge 0$ is sufficient.

\begin{lemma}\label{lem:non66}
    For all faces $f \in \planarprime$ except $6$-hexagons, we have $\ch_3(f) \ge 0$.
\end{lemma}

\begin{proof}
    In the first two steps, every face contributes at most $\frac{1}{5}$ charge per 0-edge and 1-edge resp.~to the corresponding receiving 1-triangles: 
    Wedge-receiver may receive $\frac{1}{5}$ charge, side-receiver $\alpha - \frac{1}{5}$.
    By our choice of $\alpha$, we have $2(\alpha-\frac{1}{5}) \le \frac{1}{5}$. Thus, in~ the first two steps, at most $\frac{1}{5}$ charge is contributed per edge. Therefore, any $x$-$y$-face $f$ after step~3 has
    \[\ch_3(f) \geq x+y-4-\frac{y}{5}-x\alpha\]
     charge, which is nonnegative for all faces but for $1$-triangles, $2$-triangles, $0$-quadrilaterals or $1$-quadrilaterals (recall that $\planarprime$ does not contain $0$-triangles). 

We start with the case of a 1-triangle $f$.
The two preceding propositions confirm that $f$ receives
$\frac15$ charge from its wedge-supplier and $\alpha - \frac 15$ from its side-supplier. Further it contributes $\alpha$ charge to edges. Thus, the charge in a 1-triangle $f$ is
\[\ch_3(f) \ge 3 + 1 - 4 + \frac15 + (\alpha - \frac15) - \alpha= 0.\]
A $2$-triangle $f$ does not lose any charge in step 1, and at most $4\cdot(\alpha - \frac15)$ to the side-receivers in step 2, hence 
\[\ch_3(f) \ge 2+3-4 - 4(\alpha-\frac{1}{5})-2\alpha = \frac{9}{5}-6\alpha \ge 0.\]
A $0$-quadrilateral does not lose any charge neither in steps $1$ to $3$, hence its charge remains zero.
Finally, a $1$-quadrilateral $f$ may lose $2 \cdot \frac{1}{5}$ charge in the first step and at most $2 \cdot (\alpha-\frac{1}{5})$ in the second step, hence we have 
\[\ch_3(f) \ge 1+4-4-2\cdot\frac15-2\cdot(\alpha-\frac{1}{5})-\alpha = 1 - 3 \alpha \ge 0.\qedhere\]

\end{proof}

\begin{lemma}\label{lem:66-step4}
    Let $f \in \planarprime$ be a $6$-hexagon. Then $\ch_4(f) \ge 8-30 \alpha = -2\alpha$.
\end{lemma}
\begin{proof}
    The initial charge of $f$ is $\ch(f) = 6+6-4 = 8$.
    $6$-hexagons do not contribute charge in the first two steps and lose at most $6\alpha$ in step 3.
    By 5-planarity, we have $\vert E^-(f)\vert \le 12$, because the three diagonals of $f$ in $\Gamma$ can be crossed by three additional edges each. Thus, in step 4, $f$ spends at most $12\cdot 2\alpha = 24 \alpha$. Summing up, the charge after step 4 is
    $\ch_4(f) \ge 8 - 30\alpha$.
    With our choice of $\alpha = \frac27$, it holds $8-30\alpha = -2\alpha$.
\end{proof}

\begin{lemma}\label{lem:66-HQempty}
    Let $f \in \planarprime$ be a $6$-hexagon. If $\vert E^-(f)\vert \ne 12$ or $\vert E_{HH}(f)\vert \ge 2$, then $\ch_4(f) \ge 0$. If $\vert E_{HH}(f)\vert = 1$, then  $\ch_4(f) \ge - \alpha$.
\end{lemma}

\begin{proof}
    As discussed above, $\ch_3(f) \ge 8 - 6 \alpha$. Consider $\vert E^-(f)\vert \le 11$. Thus, in step 4, $f$ loses at most $11 \cdot 2\alpha$. Summing up, the charge after step 4 is $\ch_4(f) \ge 8 - 28\alpha = 0$. Consider now that $\vert E_{HH}(f)\vert = x$. Then, in step 4, $f$ contributes at most $(12-x) \cdot 2\alpha +x \alpha= (24-x)\alpha$. Thus, $\ch_4(f) \ge (x-2) \alpha$.
\end{proof}

Next, we consider the distribution of charge in step 5, namely for the $6$-hexagons $f$ with negative charge $\ch_4(f)$. 
By \cref{lem:66-HQempty}, we have $\vert E^-(f)\vert = 12$ and $\vert E_{HH}(f)\vert \le 1$ for such a $6$-hexagon $f$. Since $\vert E_{H}(f)\vert \le 9$, it directly follows $\vert E_{HQ}(f) \vert \ge 2$.
Denote by $v_0,\dots,v_5$ the vertices of $f$ and w.l.o.g. assume that the edge $v_0v_1$ is crossed by $e \in E_{HQ}(f)$.
This implies that $v_0v_1$ is also part of a boundary of a Q-block. Let $f'$ be the face in the Q-block, to which $v_0v_1$ is incident. We call $f$ and $f'$ \emph{charging-neighbors}.
Observe that $f'$ is an $x$-$y$ face with $x \geq 2$ and $y \geq 3$ that contains an uncrossed edge (which is $v_0v_1$). Face $f'$ could be next to several H-blocks, but this always implies a unique uncrossed edge on its boundary.

\begin{lemma}
  \label{lem:no-E+}
  If a $6$-hexagon $f \in \planarprime$ with $E_{HQ}(f) \neq \emptyset$ has a charging-neighbor $f'$ and the common edge $e$ is in $E^+$, then $\ch_5(f)\ge 0$. 
\end{lemma}

\begin{proof}
    Here, the edge $e$ contributes $2\alpha$ to $f$ in step 3 if $f'$ is not a 6-hexagon.
    Otherwise, $f'$ is a 6-hexagon inside a Q-block, which may arise after deleting multiple edges in $\Gamma$ that belong to other H-blocks but have segments contained in $f'$.
    In this case, we have $\ch_4(f') \ge 8 - 6 \alpha$ and $f'$ contributes $\frac{\ch_4(f')}{6} \ge 2\alpha$ to $f$.
    Thus, $\ch_5(f) \ge 8-28\alpha = 0$.
\end{proof}

\begin{restatable}[\restateref{lem:big-faces-sufficient*}]{lemma}{lembigfacessufficient}
\label{lem:big-faces-sufficient}
If a $6$-hexagon $f \in \planarprime$ with $E_{HQ}(f) \neq \emptyset$ has a charging-neighbor $f'$ with $|f'| \geq 4$, then an even distribution of $\ch_4(f')$ over its charging-neighbors implies $\ch_5(f) \ge 0$.
\end{restatable}

The proof of this lemma can be found in the appendix.

Hence, we only focus on $6$-hexagons $f$ where all charging-neighbors $f'$ are 2-triangles or 3-triangles. In this case, one charging neighbor may not provide sufficient charge. Therefore, we first distinguish the amount of charges that can be received from different charging-neighbors.

\begin{lemma}\label{lem:charging-neighbors-contribution}
Let $f \in \planarprime$ be a $6$-hexagon with $E_{HQ}(f) \neq \emptyset$ with a charging neighbor $f'$ and $\ch_4(f)<0$.
If $f'$ is a 3-triangle, then it can contribute $\frac{2-3\alpha}{3} \ge \alpha$ charge to $f$.
If $f'$ is a 2-triangle, then it can contribute to $f$ in step~5 the following charges:
\begin{enumerate}
    \item $\frac{8}{5}-5\alpha$, if $\vert E_{HQ}(f') \vert = 1$.
    \item $\frac{7}{5}-4\alpha$, if $\vert E_{HQ}(f') \vert \ge 2$.
    \item $\frac{6}{5}-3\alpha \ge \alpha$, if an edge of $E_{HQ}(f')$ is crossed three times in the interior of $f$ in $\Gamma$.
\end{enumerate}
\end{lemma}

\begin{proof}
    First, consider that $f'$ is a 3-triangle. Then $\ch_0(f') = \ch_2(f') = 2$ and $\ch_3(f') = \ch_4(f') = 2- 3\alpha$. Distributing the remaining charge equally guarantees $\frac{2-3\alpha}{3}$ for $f$ in step 5 and this is at least $\alpha$ by the choice of $\alpha$.

    Now, we consider the case that $f'$ is a 2-triangle.
    If $\vert E_{HQ}(f') \vert = 1$, then $f'$ pays at most $3(\alpha -\frac{1}{5})$ in step 2,   
    see \cref{fig:charge-2-triag-1}. Therefore, $\ch_2(f') \ge 1-3(\alpha -\frac{1}{5})$ and $\ch_4(f') = \ch_3(f') \ge \frac{8}{5}-5\alpha.$
    If $\vert E_{HQ}(f') \vert = 2$, then $f$ pays at most $2(\alpha -\frac{1}{5})$ in step 2, because these two edges in $E_{HQ}(f')$ either cause on both sides the situation in \cref{fig:charge-2-triag-1}, see \cref{fig:charging-neighbors-1}, or at one side the situation in \cref{fig:charge-2-triag-2}.
    Therefore, $\ch_2(f') \ge 1-2(\alpha -\frac{1}{5})$ and $\ch_4(f') \ge \frac{7}{5}-4\alpha$.
    \begin{figure}[t]
    \centering
    \hspace*{\fill}
        \begin{subfigure}[b]{.24\textwidth}
		\centering
		\includegraphics[width=1\textwidth,page=5]{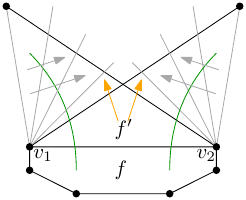}
		\subcaption{}
		\label{fig:charging-neighbors-1}
	\end{subfigure}
    \hfill
    	\begin{subfigure}[b]{.24\textwidth}
		\centering
		\includegraphics[width=1\textwidth,page=6]{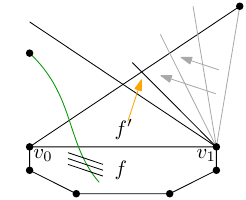}
		\subcaption{}
		\label{fig:charging-neighbors-2}
	\end{subfigure}
    \hspace*{\fill}
 	\caption{Discharging of a $2$-triangle in step 2: Cases needed for the proof of \cref{lem:charging-neighbors-contribution}.}
    \label{fig:charging-neighbors}
  \end{figure}
  
    For the last case, let $e\in E_{HQ}(f) \cap E_{HQ}(f')$ with at least three crossings in the interior of $f$ in $\Gamma$. This implies that $f$ pays at most $(\alpha -\frac{1}{5})$ in step 2, as in addition to the fact that $\vert E_{HQ}(f')\vert \ge 1$, face $f'$ has no side-neighbors at one of its edges;
    see \cref{fig:charging-neighbors-2}. Therefore, $\ch_2(f') \ge 1-(\alpha -\frac{1}{5})$ and $\ch_4(f') \ge \frac{6}{5}-3\alpha \ge \alpha$ by our choice of $\alpha$.
\end{proof}

Using the last lemma, we show that a $6$-hexagon has enough charge in all remaining~cases.

\begin{lemma}
\label{lem:charge-combined}
Let $f \in \planarprime$ be a $6$-hexagon. Then $\ch_5(f) \ge 0$. 
\end{lemma}

\begin{proof}
By \cref{lem:66-HQempty,lem:no-E+,lem:big-faces-sufficient}, we can restrict to the case that $\vert E^-(f)\vert = 12, \vert E_{HH}(f)\vert \le 1$ and all charging neighbors of $f$ are 2- or 3-triangles with the common edge not in $E^+$.

We distinguish cases by the number of \emph{charging edges} of $f$, which we define as the incident edges of $f$ in $\tilde{\Gamma}$ that would be crossed by edges of $E_{HQ}$ or $E_{HH}$ in $\Gamma$.
By our terminology, in the case of an $E_{HQ}$-charging edge, $f$ is adjacent to a charging neighbor. 
    
    If $f$ is incident to an $E_{HH}$-charging edge, we have $\vert E_{HH}(f) \vert \ge 1$ and thus $\ch_4(f)\ge -\alpha$. For the final case analysis, we call the vertices adjacent to $f$ in clockwise order $v_0, ..., v_5$.

    \begin{itemize}
        \item \emph{Case 1: $f$ is incident to exactly one charging edge.}
        W.l.o.g., let $v_0v_1$ be this charging edge. Recall that in order to establish $\vert E^-(f)\vert = 12$, all three diagonals of $f$ in $\Gamma$ must have five crossings. Consider the diagonal edge $v_2v_5$. It is crossed by the two other diagonals and at most two more edges in $E_{H}(f)$. Thus, the fifth edge is in $E_{HQ}(f)$ or $E_{HH}(f)$ and crosses $v_0v_1$, which implies it crosses another diagonal of $f$ in $\Gamma$, see \cref{fig:charge-combined-1}. This is a contradiction to $\vert E^-(f)\vert = 12$.
\begin{figure}[t]
    \begin{subfigure}[b]{.24\textwidth}
		\centering
		\includegraphics[width=1\textwidth,page=1]{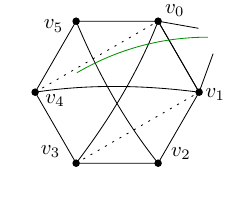}
		\subcaption{}
		\label{fig:charge-combined-1}
	\end{subfigure}
     \hfill
    	\begin{subfigure}[b]{.24\textwidth}
		\centering
		\includegraphics[width=1\textwidth,page=2]{combined-cases.pdf}
		\subcaption{}
		\label{fig:charge-combined-2}
	\end{subfigure}
     \hfill
    \begin{subfigure}[b]{.24\textwidth}
		\centering
		\includegraphics[width=1\textwidth,page=3]{combined-cases.pdf}
		\subcaption{}
		\label{fig:charge-combined-3}
	\end{subfigure}
    \hfill
    	\begin{subfigure}[b]{.24\textwidth}
		\centering
		\includegraphics[width=1\textwidth,page=4]{combined-cases.pdf}
		\subcaption{}
		\label{fig:charge-combined-4}
	\end{subfigure}
    \hfill
 	\caption{(a) Case 1, (b) case 2.1, and (c)-(d) case 2.2 of \cref{lem:charge-combined}.}
    \label{fig:charge-combined}
  \end{figure} 
        \item \emph{Case 2: $f$ is incident to exactly two charging edges.}
        If the two charging edges lie opposite, i.e., with distance 3, on the boundary of $f$, we can use the same argument as in case 1 to contradict $\vert E^-(f)\vert = 12$. For the other cases, we claim that both charging edges are crossed by edges with at least three crossings inside $f$.
        If such an edge $e\in E^-(f)$ is in $E_{HQ}(f)$, $f$ receives by \cref{lem:charging-neighbors-contribution} at least $\alpha$ charge from the corresponding charging neighbor in step 5. If $e\in E_{HH}$, then $f$ contributes only $\alpha$ (instead of $2\alpha$) charge to $e$. 
        This implies $\ch_5(f) \ge 0$ by \cref{lem:66-step4}.
        \begin{itemize}
            \item \emph{Case 2.1: The two charging edges lie neither adjacent nor opposite on the boundary of $f$.} W.l.o.g., let $v_0v_1$ and $v_2v_3$ be the charging edges. Consider the diagonal $v_1v_4$. To be crossed five times, it must have two crossings points between $v_1$ and the crossing of $v_1v_4$ and $v_0v_3$. Both of these edges start at $v_2$ and at least one of this edges crosses $v_0v_1$; we call the latter edge $e$. By symmetry, we have the same situation for the diagonal $v_2v_5$ and for this reason there exists an edge $e'$ starting at $v_1$ and crossings $v_2v_3$, see \cref{fig:charge-combined-2}. Now, both $e$ and $e'$ are crossed three times in $f$.
            \item \emph{Case 2.2: The two charging edges lie adjacent on the boundary of $f$.}
            W.l.o.g., let $v_0v_1$ and $v_1v_2$ be the charging edges. With the same arguments as above, we can establish the existence of an edge starting at $v_5$ and crossing $v_0v_1$ as well as an edge $e$ starting at $v_3$ and crossing $v_1v_2$.            
            Further, there must exist two edges crossing $v_1v_4$ between $v_1$ and the crossing point with $v_0v_3$; w.l.o.g., one of this edges, called $e'$, starts at $v_2$ and crosses $v_0v_1$, while the second edge is called $e''$.
            Note that there must be a second edge besides $e$ that crosses $v_2v_5$ between $v_2$ and the crossing point with $v_0v_3$.
            
            If $e''$ starts at $v_2$, then $e$ and $e'$ are crossed three times inside $f$. Otherwise $e''$ starts at $v_0$ and crosses $v_1v_2$, so both $e'$ and $e''$ are crossed three times.            
        \end{itemize}
        
        \item \emph{Case 3: $f$ is incident to three or more charging edges.} From each charging \textit{neighbor}, $f$ receives by \cref{lem:charging-neighbors-contribution} at least $\frac{8}{5}-5\alpha$ charge in step 5. If $\vert E_{HH}(f)\vert = 1$, then $\ch_5(f) \ge - \alpha + 2 \cdot (\frac{8}{5}-5\alpha) \ge 0$. If $\vert E_{HH}(f)\vert = 0$, then the only possible case for $\ch_5(f) < 0$ is (using again \cref{lem:charging-neighbors-contribution}) that $f$ is incident to exactly three charging neighbors that are all 2-triangles in which one $E_{HQ}$ ends in $\Gamma$. This implies $\vert E_{H} \vert = 9$ and thus all edges of $E_{HQ}(f)$ are crossed three times in the interior of $f$ in $\Gamma$. Therefore $\ch_5(f) \ge - 2 \alpha + 3 (\frac{6}{5}-3\alpha) \ge 0$.\qedhere
    \end{itemize}
\end{proof}

This concludes the proof of \cref{thm:5planar}, as all faces and edges have enough charge. $G$ has at most $7(n-2)$ edges.
\end{proof}

This new bound can be used to improve the crossing lemma.

\begin{theorem}\label{lem:cross-lemma}
Let $G$ be a graph with $n$ vertices and  $m \geq 7.39n$ edges. Then
$cr(G) \geq \frac{1}{27.3}\frac{m^3}{n^2}$.
\end{theorem}
\begin{proof}
    \label{thm:crossing_lemma}
    Analogously to the proof of \cref{lem:outer-crossing-lemma} we find that

    \[cr(G) \ge 6(m-7(n-2)) + \left(5 \cdot 7(n-2) - \frac{203}{9}(n-2)\right) = 6m - \frac{266}{9}(n-2).\]

    The second term is a bound derived in \cite{BestBestConstantCr}, slightly improving the conference version \cite{BestConstantCr}.

    Choose $p = \frac{133n}{18m}\le 1$ and proceed with the probabilistic proof to obtain
    \[cr(G) \ge \frac{6m}{p^2} - \frac{266n}{9p^3} = \frac{648}{17689} \frac{m^3}{n^2} \ge \frac{1}{27.3} \frac{m^3}{n^2}.\qedhere\]
\end{proof}

Combining the new result with $cr(G) \le \frac{k}{2} \cdot m$, which holds for $k$-planar graphs, yields
\begin{corollary}\label{cor:density-bound}
An $n$-vertex simple topological $k$-planar graph has at most $3.7\sqrt{k}n$ edges.
\end{corollary}

\section{Further applications of our technique}\label{sec:further-appl}

\subsection{k-planar graphs, in particular 6-planar graphs}
We generalize our technique to arbitrary $k \geq 6$ and obtain the following result.
\begin{theorem}\label{thm:one_point_five}
An $n$-vertex simple topological $k$-planar graph with $k \geq 6$ has at most $1.5k(n-2)$ edges.
\end{theorem}

\begin{proof}
We use the same charging technique and notation as in the proof of \cref{thm:5planar}, however we do not distribute charge from Q-blocks to H-blocks. Fix $\alpha = \frac{2}{1.5k}$. 
 
\textbf{H-blocks.}
Fix an H-block, i.e., a 6-hexagon $f$. Since $G$ is $k$-planar, at most $3(k-1)$ edges intersect the interior of $f$.
Hence, including the boundary edges, $f$ requires at most $((3(k-1))\cdot 2+ 6)\alpha = 6k\alpha = 8$ charge. Since $f$ was charged with 8 charge, $\ch_5(f) \ge 0$.
  
\textbf{Q-blocks.}
In the first of two redistribution steps, every $1$-triangle obtains $\frac{1}{5}$ from its wedge-supplier. Afterwards, $1$-triangles require an additional charge of $(\alpha-\frac{1}{5})$, which it obtains through its side-suppliers.
Observe that for $k\geq 7$, we have $\alpha < \frac{1}{5}$ and the second step is obsolete. 
Assume that $k=6$, which implies that every $1$-triangle requires an additional charge of $\frac{1}{45}$. Since our graph is $6$-planar, a face loses at most $4 \cdot \frac{1}{45}$ of charge through every incident 1-edge.
Since $\frac{4}{45} < \frac{1}{5}$, we can once again conclude that  $\ch_2(f) \geq x+y-4-\frac{y}{5}-x\alpha$ holds for every face. This is nonnegative unless $f$ is a
$0$-quadrilateral, a $2$-triangle, or a $1$-quadrilateral. We finish the proof by observing that a $0$-quadrilateral does not lose charge in either step, hence its charge remains zero. A $1$-quadrilateral may lose $\frac{2}{5}$ charge in the first step and at most $4(\alpha-\frac{1}{5})$ in the second step. Since $1 - \frac{2}{5} - \alpha -  4(\alpha-\frac{1}{5}) = \frac{7}{5} - 5\alpha = \frac{7}{5}-\frac{10}{9}\geq 0$, it has sufficient charge.
Finally, a $2$-triangle does not lose charge in the first step and since it loses at most $6(\alpha-\frac{1}{5})$ charge in the second step we have
$1 - 6(\alpha-\frac{1}{5})-2\alpha = \frac{9}{5}-8\alpha = \frac{9}{5}-\frac{16}{9} > 0$.
\end{proof}
This strengthens the bound of $3(k+1)(n-2)$, which can be simply derived by the thickness of $k$-planar graphs. For $6$-planar graphs, \cref{thm:one_point_five} yields a bound of $9(n-2)$ edges for simple topological  graphs, a small improvement over the best bound of $\approx 9.06n$ (\cref{cor:density-bound}).

\begin{corollary}
 An $n$-vertex outer $6$-planar graph has at most $5n-10$ edges.
 \end{corollary}
 \begin{proof}
Let $\alpha = \frac{2}{9}$ and let $G$ be an $n$-vertex outer $6$-planar graph with maximum edge density.
We apply the same charging technique as in the proof of \cref{thm:one_point_five}. In particular, let $\Gamma$ be a fixed outer $6$-planar drawing of $G$. Assign $|f|+|v(f)|-4$ charge to every face of $\planar$. Observe that the outer face obtains $2n-4$ charge. We ensure $2\alpha$ charge for every edge, while the charge of the faces remains nonnegative.
The outer face distributes exactly $\alpha$ charge to each of the $n$ boundary edges, the remainder of the charge is simply subtracted from the total. Hence, $4n-8-(2n-4-\alpha n) \geq 2\alpha m$, which implies $m \le 5(n-2)$.
\end{proof}

For a lower bound, choose $n = 5x+2$ for some integer $x\geq 1$, split an $n$-gon into $x$ faces of size $7$ and insert all $14$ chords into every face, which is $6$-planar. This immediately gives
  
\begin{theorem}
There exist outer $6$-planar graphs with $n$ vertices and $4(n-2)$ edges.
\end{theorem}
By copying all interior edges of this construction to the outer face, we obtain
\begin{theorem}
There exist $6$-planar multi-graphs with $n$ vertices and $7(n-2)$ edges.
\end{theorem}

\begin{theorem}
There exist simple $6$-planar n-vertex graphs with $6.8n - 22.8$ edges.
\end{theorem}
\begin{proof}
Start with a hexagonal tiling $H_x$ of the surface of a cylinder (cmp. \cref{thm:simple-framed-lb}) for some even $x$. Add $\frac{3x}{2}$ vertices onto edges incident to two hexagons so that all $3x$ faces on the lateral surface (i.e. all faces except the top and bottom) are 7-gons. This graph has $\frac{15}{2}x + 6$ vertices, $\frac{21}{2}x+6$ edges, $3x$ 7-gons and two hexagons.
 
Fill each 7-gon with $K_7$ and the hexagons with $K_6$. There is one multi-edge at every vertex we added to $H_x$ (so $\frac{3}{2} x$ many) as well as 3 multi-edges in the top and bottom face each. Remove them to get a simple $6$-planar graph with $\frac{21}{2} x + 3x \cdot 14 + 2\cdot 9 - \frac{3}{2}x = 51x + 18$ edges, which yields $6.8n - 22.8$ for $x = \frac{2}{15} \cdot (n-6)$.
\end{proof}
 
\subsection{A short and simple proof for 4-planar graphs}

Finally, using the same technique, we achieve for 4-planar graphs a surprisingly good result compared to the tight bound of $6(n-2)$ \cite{DBLP:journals/comgeo/Ackerman19}, which requires a far more complex proof.

\begin{theorem}\label{thm:four-planar}
An $n$-vertex simple topological $4$-planar graph has at most $6.25(n-2)$ edges.
\end{theorem}

\begin{proof}
    We adapt the proof strategy from \cref{thm:one_point_five}. Fix $\alpha = \frac{2}{6.25} = \frac{8}{25}$.
    
    By 4-planarity, an H-block contains at most nine edges. Thus, a 6-hexagon has sufficient charge as it requires not more than $(9\cdot 2 + 6) \alpha \le 8$ charge.
    In a Q-block, 1-triangles receive $\frac{1}{5}$ charge in the first step and require $\frac{3}{25}$ more in the second step. Again, the only critical faces with potentially a negative charge are 1-quadrilaterals and 2-triangles. A 1-quadrilateral may lose $\frac{2}{5}$ charge in the first step and $2\cdot \frac{3}{25}$ in the second step, and therefore it has enough charge as $1-\frac{2}{5}-\frac{6}{25}-\alpha \ge 0$.
    A 2-triangle does not contribute charge in the first step and by 4-planarity at most $3 \cdot \frac{3}{25}$ in the second step. It may need all of its charge as $1-\frac{9}{25}-2\alpha = 0$, but still has a non-negative final charge.
\end{proof}

For a recent refinement of this approach achieving the optimal bound $6(n-2)$, see \cite{bungener2025simplified}.

\section{Conclusions and Open Problems}\label{sec:open-problems}
We have considered density bounds for the class of 5-planar graphs. New techniques for the analysis of non-planar graph structures led to various new insights in structural properties of $k$-planar graphs for $k \geq 5$.
Our work now leads to a number of open questions:

\begin{itemize}
    \item Refine the bounds for (outer) $5$-planar graphs. We suspect that simple $5$-planar graphs have at most $6.2n-\mathcal{O}(1)$ edges, while there exist non-homotopic $5$-planar multigraphs with $6.4n-O(1)$ edges. These can be obtained by duplicating all chords of the construction of \cref{thm:lower-outer} into the outer face. An improvement of the upper bound to $6.2n-\mathcal{O}(1)$
    would also improve the constant of the Crossing Lemma to $\frac{1}{25.85}$.
    \item Extend our technique to non-simple $k$-planar drawings and graphs.
    \item Apply our technique to other graph classes such as $k$-planar $C_r$-free graphs in order to extend the work of~\cite{BekosBBDHKMOW25} for $r=3,4$. 
    \item Unlike for $k \leq 4$, polyhedral $h$-framed $5$-planar graphs fall short of achieving the density of optimal $5$-planar graphs by a linear margin. Is this true for any larger $k \geq 5$?
\end{itemize}
\noindent

\bibliographystyle{plainurl}
\bibliography{bib2doi}

\newpage
\appendix
\section{Omitted proofs from the main part}

\lembigfacessufficient*
\label{lem:big-faces-sufficient*}

\begin{proof}
Assume the $x$-$y$-face $f'$ is charging-neighbor to $k$ many $6$-hexagons.
Assume first that $k = 1$. Then, $x \geq 2$.
If $y = 4$, then $f'$ loses at most $\frac{1}{5}$ charge in step 1 and at most $3 \cdot (\alpha-\frac{1}{5})$ charge in step 2, see \cref{fig:charge-2-triag-1}.
Hence, it has 
\[\ch_4(f') \ge 2-\frac{1}{5}-3(\alpha-\frac{1}{5}) - 2\alpha = \frac{12}{5}-5\alpha\]
charge, which is contributed to $f$ in step 5. Thus,
\[\ch_5(f) \ge 8 -30 \alpha + (\frac{12}{5}-5\alpha) \ge 0.\]
For $y > 4$, we have $\frac{4y}{5} \geq 4$, hence
$\ch_4(f') \geq x+y-4-\frac{y}{5}-x\alpha \geq (1-\alpha)x \geq 2-2\alpha$ (since $x \geq 2$). Thus,
\[\ch_5(f) \ge 8-30\alpha + (2 - 2\alpha) \ge 0.\]

For $k > 1$, observe that $x > 2$, $x \geq k$, and $y \geq k$ has to hold. Moreover, since we did not lose any charge over uncrossed edges in neither step 1 or 2, we get 
\[\ch_3(f') \geq x+y-4-\frac{y-k}{5}-x\alpha.\]
To prove the lemma, we have to show that this is at least $k \cdot (30\alpha - 8)$.
For $y \ge 5$, we directly see 
\[\ch_3(f') \ge  k(\frac{6}{5}-\alpha) - 4+ \frac{4}{5}y \ge k(\frac{6}{5}-\alpha) \ge  k \cdot (30\alpha - 8)\] using $x \ge k$ and plugging in our choice of $\alpha$.
Assume now that $y=4$ and $1 < k \leq 4$. Observe that $k=2$ implies $x=3$, while $k \ge 3$ implies $x=4$. If $k=2$, then $\ch_4(f') \ge 3 - \frac{2}{5} - 3\alpha \ge 2\cdot(30\alpha-8)$. For $k=3$ and $k=4$, no charge is distributed in step 2 and $\ch_4(f') \ge 4 - 4 \alpha \ge 4 \cdot (30\alpha-8)$ holds.
Thus, we have shown that $f'$ can contribute in all cases enough charge in step 5 to $f$ (and to potentially more $6$-hexagons).
\end{proof}

\end{document}